\documentclass[twocolumn]{IEEEtran} 

\usepackage{amsmath,amssymb,amsfonts}
\usepackage[final]{graphicx}
\usepackage{psfrag}
\usepackage{mwe}
\usepackage[tight]{subfigure}
\usepackage[numbers,sort&compress]{natbib}
\usepackage{multicol}
\usepackage{mathrsfs} 
\usepackage{epstopdf}
\usepackage{multirow}
\usepackage{amsthm,mathtools}
\usepackage{float}
\usepackage{booktabs}  
\usepackage{multirow} 
\usepackage{relsize}
\usepackage{bm}
\usepackage{hyperref}
\usepackage{tabularx}
\usepackage[nolist,nohyperlinks]{acronym}
\usepackage{svg}
\usepackage{xparse}
\NewDocumentCommand{\evalat}{sO{\big}mm}{%
  \IfBooleanTF{#1}
   {\mleft. #3 \mright|_{#4}}
   {#3#2|_{#4}}%
}

\begin{acronym}
\acro{PLS}{physical layer security}
\acro{OP}{outage probability}
\acro{SNR}{signal-to-noise-ratio}
\acrodefplural{SNR}[SNRs]{signal-to-noise ratios}
\acro{AWGN}{additive white Gaussian noise}
\acro{CSI}{channel state information}
\acro{PDF}{probability density function}
\acrodefplural{PDF}[PDFs]{probability density functions}
\acro{CDF}{cumulative distribution function}
\acro{PDF}{probability density function}
\acrodefplural{CDF}[CDFs]{cumulative distribution functions}
\acro{LRS}{large reflecting surface}
\acro{BS}{base station}
\acro{RF}{radio frequency}
\acro{RV}{random variable}
\acro{FN}{folded normal}
\acro{TDD}{time-division duplexing}
\acro{LOS}{line-of-sight}
\acro{DL}{downlink}
\acro{MRT}{maximal ratio transmission}
\acro{MRC}{maximal ratio combining}
\acro{RIS}{reconfigurable intelligent surface}
\acro{MC}{Monte Carlo}
\acro{RV}{random variable}
\acro{CDF}{cumulative density function}
\acro{FAS}{fluid antenna system}
\acro{MIMO}{multipe-input multiple-output}
\acro{SOP}{secrecy outage probability}
\acro{BS}{base station}
\acro{NOMA}{non orthogonal multiple access}
\acro{mmWave}{millimeter wave}
\acro{BER}{bit error rate}
\acro{CSI}{channel state information}
\acro{UE}{user equipment}
\end{acronym}

\newtheorem{proposition}{Proposition}

\newtheorem{remark}{Remark}

\usepackage[numbers,sort&compress]{natbib}

\makeatletter
\def\blfootnote{\xdef\@thefnmark{}\@footnotetext}
\makeatother

\usepackage{float}
\usepackage{stfloats}
\usepackage{textcomp}
\usepackage[ruled,vlined]{algorithm2e}

\begin{document}
\title{\huge{Reliable and Secure Communications Through Compact Ultra-Massive Antenna Arrays}}
\author{Jos\'e~David~Vega-S\'anchez, \textit{Member, IEEE}, Henry Ramiro Carvajal Mora, \textit{Senior Member, IEEE}, Nathaly Ver\'onica Orozco Garz\'on, \textit{Senior Member, IEEE}, and F.~J.~L\'opez-Mart\'inez, \textit{Senior Member, IEEE} }

\maketitle


\blfootnote{\noindent 
J.~D.~Vega-S\'anchez is with the Colegio de Ciencias e Ingenier\'ias  ``El Polit\'ecnico", Universidad San Francisco de Quito (USFQ), Diego de Robles S/N, Quito 170157. (e-mail: dvega@usfq.edu.ec).
}

\blfootnote{\noindent
H.R.C. Mora, and N.V.O. Garz\'on are with the Faculty of Engineering and Applied Sciences, Telecommunications Engineering, Universidad de Las Américas (UDLA), 170503 Quito, Ecuador, e-mails: $\{\text{henry.carvajal,nathaly.orozco}\}$@udla.edu.ec.
}

\blfootnote{\noindent
F.~J.~L\'opez-Mart\'inez is with the Department of Signal Theory, Networking
and Communications, Research Centre for Information and Communication
Technologies (CITIC-UGR), University of Granada, 18071, Granada (Spain). (e-mail: fjlm@ugr.es). 
}

\blfootnote{\noindent
The work of H.R.C. Mora and N.V.O. Garzón was supported by Universidad de Las Américas, under research project ERT.HCM.23.13.01. This work is supported by grant PID2023-149975OB-I00 (COSTUME) funded by MICIU/AEI/10.13039/501100011033.
}
\blfootnote{\noindent Corresponding author: Henry Ramiro Carvajal Mora (e-mail: henry.carvajal@udla.edu.ec).}

\begin{abstract}
Compact Ultramassive Antenna Array (CUMA) is a pioneering paradigm that leverages the flexibility of the Fluid Antenna System (FAS) to enable a simple multiple access scheme for massive connectivity without the need for precoding, power control at the base station or interference mitigation in each user's equipment. In order to overcome the mathematical intricacy required to analyze their performance, we use an asymptotic matching approach to relax such complexity with a remarkable accuracy. First, we analyze the performance of the CUMA network in terms of the outage probability (OP) and the ergodic rate (ER), deriving simple and highly accurate closed-form approximations to the channel statistics. Then, we evaluate the potential of the CUMA scheme to provide secure multi-user communications from a physical layer security perspective. Leveraging a tight approximation to the signal-to-interference-ratio (SIR) distribution, we derive closed-form expressions for the secrecy outage probability (SOP). We observe that the baseline CUMA (without side information processing) exhibits limited performance when eavesdroppers are equipped with a CUMA of the same type. To improve their secure performance, we suggest that a simple imperfect interference cancellation mechanism at the legitimate receiver may substantially increase the secrecy performance. Monte Carlo simulations validate our approximations and demonstrate their accuracy under different CUMA-based scenarios.
\end{abstract}


\section{Introduction}
The forthcoming sixth-generation (6G) network is expected to support enormously massive connectivity of devices in extreme-density and diverse mobility scenarios. To achieve this goal, a simple multiple access technology that enables a massive number of connected user equipments (UEs) at the same time-frequency is of paramount importance. In this context, non-orthogonal multiple access (NOMA) and rate splitting multiple access (RSMA) have emerged as aggressive approaches to massive connectivity for multi-user signal handling~\cite{Bruno}. However, just like massive \ac{MIMO} in fifth-generation (5G), RSMA and NOMA require channel state information at the \ac{BS} for precoding and power control optimization, as well as for the UE to perform complex interference cancellation techniques \cite{nomaa}. To circumvent this issue, \ac{FAS} has been envisioned as a flexible structure to enable position-switchable antennas, allowing upscaling connectivity in a simpler and more scalable way~\cite{TongNewPAra}. Specifically, FAS uses flexible and reconfigurable antennas (e.g., liquid metals, or pixel-like switches) to {create} a software-controllable structure equipped with positions (also called ``ports") that can be adaptively modified in space at the UE side for exploiting its dynamic position over the available space at the UE to pick up the strongest or least-interfered received signal~\cite{Ref2}. 

The promising usefulness of FAS opens up many use cases for beyond fifth-generation (B5G) networks: $i)$ FAS as an approximation to continuous-aperture \ac{MIMO} (CP-MIMO), $ii)$ fluid antenna multiple access (FAMA) for massive connectivity, $iii)$ FAS for wireless power transfer systems, and $iv)$ FAS for \ac{PLS}, among others~\cite{ByoungRO}. For multi-user communications, fluid antenna multiple access (FAMA) facilitates a simple multiple-access method for massive connectivity. Specifically, UEs occupy the same time-frequency channel without precoding and power control at the BS or interference mitigation techniques at the UE side. Besides, each UE uses the fluid antenna to switch to the port with the strongest signal-to-interference ratio SIR, so FAMA does not require coordination between the BS and the UEs~\cite{AccessFAS}. Following this, 
in~\cite{WongfFAMA}, Wong et al. unleashed a radical approach to massive connectivity of tens or even hundreds of UEs, referred to as fast FAMA ($f$-FAMA). Here, a single fluid antenna at each UE exploits the interference null, produced inherently by multipath propagation for multiple access. However, the main challenge of $f$-FAMA is that each UE is required to identify the best port (i.e., the port with the lowest SIR) on its own by switching on a symbol-by-symbol, which is quite complicated to reach. Aiming to overcome the practical constraints of $f$-FAMA, in~\cite{WongsFAMA}, Wong et al. introduced a more realistic multiple access procedure referred to as slow FAMA ($s$-FAMA), in which the UE updates its best port only in instances where the channels' envelopes change. Therein, the interference immunity of $s$-FAMA was explored by analyzing the outage probability. Based on~\cite{WongsFAMA}, in~\cite{sFAMArefXL}, Wong et al. investigated the multiplexing gain performance of $s$-FAMA in the millimeter-wave band by proposing a channel model that can characterize the presence of directional line-of-sight (LoS) and non-LoS paths via computer simulations. In~\cite{NoorsFAMA}, Waqar et al. proposed a low-complexity port selection method for $s$-FAMA so that by employing deep learning, only a small number of correlated SINR observations are required to reach maximal multiplexing gain performance.

 Keeping the above in mind, very recently, in~\cite{WongCuma}, Wong et al. proposed an enhanced version of $s$-FAMA technology by introducing more degrees of freedom to the port selection, i.e., the UE activates many ports for the reception. Specifically, the activated ports are selected to guarantee that the in-phase and quadrature components of the desired signal at the ports are added constructively while the interference signals superimpose randomly. This scheme is called compact ultra-massive antenna array (CUMA)\footnote{The term ``ultra massive antenna array" refers to the enormous number of ports on the UE rather than the antenna port size. Interested readers can refer to~\cite{Ref3} for a comprehensive review of how FAS works and its possible use cases, including FAS-enabled massive connectivity. Specifically,~\cite[Sec. III-C]{WongCuma} explains in detail the multi-port activation scheme for determining the set of ports for each UE in CUMA.}, which can be implemented by deploying a FAS at the UE. While exhibiting promising features to enable multi-user communications beyond $s$-FAMA, its inherent mathematical complexity poses some challenges to the performance analysis of CUMA networks.
 
  From a PLS perspective, the use cases of FAS mostly focus a single-user scenario. For instance, Tang et al. proposed a PLS scheme for FAS where the BS uses one antenna to transmit information and another to send artificial noise to disrupt potential eavesdroppers~\cite{PLS1}. Since the UE is equipped with a FAS, interference can be avoided by activating the proper port. Besides, in~\cite{PLS2Far}, Ghadi et al. investigated the secrecy performance via the copula method by assuming a FAS-aided communication. Likewise, Sánchez et al., in \cite{Sanchez}, explored the secrecy performance of a FAS setup undergoing Nakagami-$m$ fading channels. Here, two different implementations of FAS were considered: $i)$ a non-diversity FAS and maximum-gain combining-FAS (MGC-FAS) diversity scheme at the legitimate UE node and $ii)$ multiple antennas performing maximal ratio combining (MRC) at the eavesdropper. However, the question of whether FAS can provide secure communications in combination with its ability to serve multiple users remains unanswered to date.

 Based on the above and encouraged by the potential of CUMA to provide a simple multiple-access scheme, we exploit the benefits of asymptotic techniques to approximate CUMA's statistics through a simple approach without invoking complicated single/double-fold integrals, like in previous works \cite{WongCuma}. Then, with these results at hand, we provide analytically tractable expressions for several performance metrics, e.g., outage probability (OP) and the ergodic rate (ER) of the conventional CUMA network in terms of well-known functions available in popular computing software packages. Then, aiming to discern whether CUMA can provide additional gains in terms of PLS for forthcoming wireless networks, we explore the secrecy performance of the CUMA-aided secure communications for the first time in the literature. Again, this is accomplished in a simple way without resorting to cumbersome methods commonly used when dealing with FAS-enabled communications. The secrecy performance of a baseline CUMA scheme is evaluated, and compared to a simple but more advanced scheme that implements partial interference cancellation (IC) for CUMA. We observe that the use of a simple imperfect IC scheme can improve the security of CUMA networks. 
 In summary, the main contributions of this work are enumerated as follows: 
\begin{figure*}[ht!]
\psfrag{B}[Bc][Bc][0.7]{$\mathrm{UE~2_\mathrm{B}}$}
\psfrag{A}[Bc][Bc][0.7]{$\mathrm{UE}~1_\mathrm{B}$}
\psfrag{C}[Bc][Bc][0.7]{$\mathrm{UE}~U_\mathrm{B}$}
\psfrag{L}[Bc][Bc][0.7]{$\mathrm{Base~Station}$}
\psfrag{K}[Bc][Bc][0.7]{$N_2~\mathrm{ports}$}
\psfrag{J}[Bc][Bc][0.7]{$\lambda W_2~\mathrm{length}$}
\psfrag{P}[Bc][Bc][0.7][90]{$N_1~\mathrm{ports}$}
\psfrag{H}[Bc][Bc][0.7][90]{$\lambda W_1~\mathrm{length}$}
\psfrag{Z}[Bc][Bc][0.7]{$\mathrm{Combination}$~$\mathrm{activated}$}
\psfrag{R}[Bc][Bc][0.7]{$\mathrm{port}$}
    \subfigure[] {\includegraphics[width=0.40\textwidth]{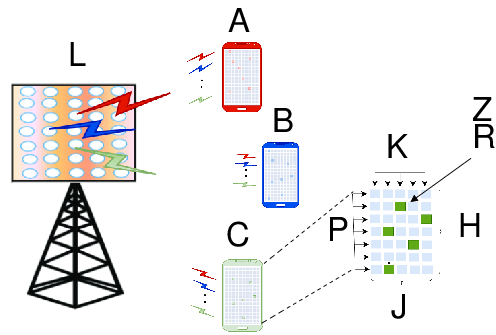}} 
\psfrag{B}[Bc][Bc][0.7]{$\mathrm{UE~1_\mathrm{E}}$}
\psfrag{Q}[Bc][Bc][0.7]{$\mathrm{UE}~U_\mathrm{E}$}
\psfrag{A}[Bc][Bc][0.7]{$\mathrm{UE}~1_\mathrm{B}$}
\psfrag{C}[Bc][Bc][0.7]{$\mathrm{UE}~U_\mathrm{B}$}
\psfrag{L}[Bc][Bc][0.7]{$\mathrm{Base~Station}$}
\psfrag{K}[Bc][Bc][0.7]{$N_2~\mathrm{ports}$}
\psfrag{J}[Bc][Bc][0.7]{$\lambda W_2~\mathrm{length}$}
\psfrag{P}[Bc][Bc][0.7][90]{$N_1~\mathrm{ports}$}
\psfrag{H}[Bc][Bc][0.7][90]{$\lambda W_1~\mathrm{length}$}
\psfrag{Z}[Bc][Bc][0.7]{$\mathrm{activated}$}
\psfrag{R}[Bc][Bc][0.7]{$\mathrm{port}$}
  \hspace{15mm}
    \subfigure[]
    {\includegraphics[width=0.40\textwidth]{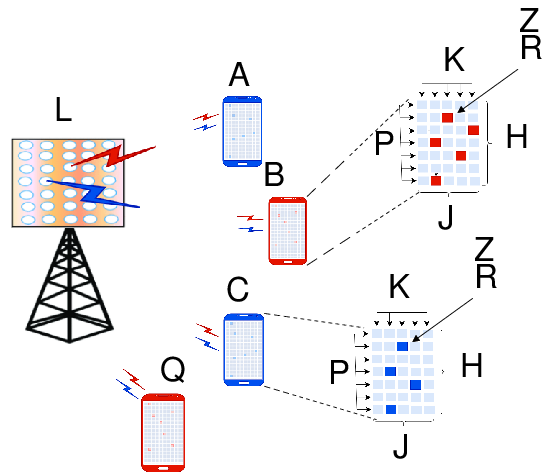}}
    \caption{(a) Conventional CUMA downlink system with a MIMO BS communicating with the $U_\mathrm{B}$ FAS-enabled UEs, (b) a CUMA-aided secure scheme with a MIMO BS communicating with  $U_\mathrm{B}$ FAS-enabled legitimate UEs in the presence of $U_\mathrm{E}$ FAS-enabled eavesdropper UEs. }
    \label{SM}
    \vspace{-4mm}
\end{figure*}
\begin{itemize}
\item We provide a closed-form approximation to the \ac{PDF} of the SIR for the in-phase channel component associated with a typical UE of the CUMA scheme. It is worthwhile noting that the proposed \ac{PDF} enables for simple performance examination of CUMA applications, which involves both the \ac{PDF} and the \ac{CDF} distributions, as well as secrecy/outage- and ergodic-based metrics. Furthermore, unlike state-of-the-art CUMA-related works, our closed-form expressions can be calculated without incurring intricate special functions or complex single/multi-fold integrals, thus demanding negligible computational cost. 

\item Based on the \ac{PDF} of the SIR for the in-phase channel component, a simple and highly accurate closed-form approximation of the \ac{PDF} and \ac{CDF} for the overall $ \mathrm{SIR}$ of the CUMA scheme is derived by taking into consideration both in-phase and quadrature components of the desired signal. Additionally, these analytical approximations can be quickly calculated employing standard software packages, as they do not involve complex functions or multi-fold integrals. 

\item 
A closed-form lower bound expression for the \ac{SOP} in the proposed CUMA-based scheme is derived by assuming a practical setup where the \ac{CSI} of the eavesdropper is not available at the \ac{BS}. It is worth emphasizing that our approximations help design secure systems from a CUMA perspective since they enable us to comprehend how the secrecy performance changes by varying the CUMA parameters. 

\item  Based on the derived formulations, valuable insights are provided on the impact of the system parameters (e.g., port density, the number of users, and the antenna size of the \ac{FAS}) over the secrecy performance.
 \end{itemize}


\section{SYSTEM MODEL AND DEFINITIONS}
Here, we assume two CUMA-based scenarios, namely, $i)$ the conventional CUMA network: a downlink setup consisting of a MIMO BS with $N_t$ antennas communicating only with legitimate $U_\mathrm{B}$ CUMA UEs, as shown in Fig. \ref{SM}a, and $ii)$ the CUMA-aided secure network: a downlink setup with a MIMO BS communicating with  $U_\mathrm{B}$ FAS-enabled legitimate UEs in the presence of $U_\mathrm{E}$ FAS-enabled eavesdropper UEs, as depicted in Fig. \ref{SM}b. For the sake of compactness, and given that the conventional CUMA scenario can be seen as a special case of the CUMA-aided secure scheme on which no eavesdroppers are present, the signal model and subsequent mathematical derivations will be developed having the latter case in mind. Hence, in the sequel we will consider a downlink wiretap network with a BS named as Alice (A) equipped with $N_t$ antennas, aiming to serve $U_\mathrm{B}=U_\mathrm{E}=U$ UEs (with $N_t\geq U$). The BS uses random precoding \cite{WongCuma} to send a confidential message to each legitimate CUMA UE Bob (B), denoted by $u_{\mathrm{B},u}$, while one external eavesdropper Eve (E) also equipped with a CUMA, denoted by $u_{\mathrm{E},u}$, tries to decode the information transmitted to each legitimate $u-$th UE from its received signal (i.e., $u \in \left \{1,\cdots,U  \right \}$), as shown in Fig. \ref{SM}b. 

Unlike CUMA-aided secure scheme, for the conventional CUMA scenario, eavesdroppers, i.e., $u_{\mathrm{E},u}$ UEs are not considered, but only legitimate $u_{\mathrm{B},u}$ UEs. By having CUMA technology on the receiving nodes, communication occurs in the same time-frequency resource, and the antennas at A are far enough apart to be considered spatially independent. \textcolor{black}{As we will later see, in an interference-limited CUMA scenario, all UEs are independent and identically distributed\footnote{We would like to highlight that this is not an assumption, but a direct consequence of the CUMA operation as the desired and interfering signals have the same average power before port selection is carried out.} (i.i.d.), which is enough to focus on the performance of any typical UE in the CUMA network}. Hence, the sub-index $u$ is dropped in the sequel for the sake of notational simplicity. UEs $u_\mathrm{B}$ and Eves $u_\mathrm{E}$ are equipped with two-dimensional (2D) FAS moving freely along $N=N_1\times N_2$ positions (ports) equally distributed on the physical space $W=W_1\lambda\times W_2\lambda$ (see Fig. \ref{SM}), where $\lambda$ is the wavelength. 

In this setup, the received signals at $u_i$ with $i \in \left \{ \mathrm{B}, \mathrm{E} \right \}$ in a vector form (i.e., at each FA port) are given by 
\cite{WongCuma}
\begin{align}\label{cumaeq1}
{\bf{r}}^{(u_i)}= {\bf{g}}^{(u_\mathrm{B},u_i)}s_{u_\mathrm{B}}+\sum_{\substack{\widetilde{u}_\mathrm{B}=1\\ \widetilde{u}_\mathrm{B}\neq u_\mathrm{B}}}^{U}{\bf{g}}^{(\widetilde{u}_\mathrm{B},u_i)}s_{\widetilde{u}_\mathrm{B}}+{\bf{\eta}}^{(u_i)},  
\end{align}
where $s_{u_\mathrm{B}}$ is the transmitted symbol for UE $u_\mathrm{B}$, $s_{\widetilde{u}_\mathrm{B}}$ is the interfering symbol produced by the transmission towards the legitimate UE $\widetilde{u}_\mathrm{B}$ in the same time-frequency resource\footnote{For simplicity, we assume that all symbols have the same power $\sigma_s^2$.},  ${\bf{\eta}}^{(u_i)}$ is the complex additive white Gaussian noise (AWGN) vector at the ports for UE $u_i$, and ${\bf{g}}^{(u_\mathrm{B},u_i)} \stackrel{\Delta}{=}{\bf{H}}_{u_i}{\bf{b}}_{u_\mathrm{B}}$ denotes the effective channel vector between the BS and the ports of UE $u_i$. Matrix ${\bf{H}}_{u_i}\in \mathbb{C}^{N\times N_t}$ denotes the complex channel between the multiantenna
BS and the FA ports of UE $u_i$, \textcolor{black}{which can also include the effects of mutual coupling between activated ports as described in \cite[Sec. III-E]{WongCuma};} ${\bf{b}}_{u_\mathrm{B}}\in\mathbb{C}^{N_t}$ denotes a random beamforming vector as defined in \cite{WongCuma} used to transmit the symbol $s_{u_\mathrm{B}}$ along the $N_t$ transmit antennas, agnostic to channel state information (CSI). Also, ${\bf{g}}^{(\widetilde{u}_\mathrm{B},u_i)} \stackrel{\Delta}{=}{\bf{H}}_{u_i}{\bf{b}}_{\widetilde{u}_\mathrm{B}}$ is the channel vector associated to the random beamforming vector towards $\widetilde{u}_\mathrm{B}$ observed at the ports of $u_i$.

The main idea of CUMA is that each FAS-enabled UE $u_i$ can activate a certain number of ports out of the available $N$, which are stored in two sets $\mathcal{K}_I$ and $\mathcal{K}_Q$ of size $|\mathcal{K}|=|\mathcal{K}_I|+|\mathcal{K}_Q|$ to maximize the in-phase ($\mathrm{I}$) and quadrature ($\mathrm{Q}$) channel components for signal reception. Specifically, the doubling-down CUMA version \cite[Sec. III-D]{WongCuma} ensures that the resulting signal in the CUMA UE  $u_i$ is aligned both in-phase and quadrature to subsequently aggregate them for the detection method (e.g., matrix inverse approach). Therefore, the resulting signal at UE $u_i$ is $r^{u_i}=r^{u_i}_\mathrm{I}+r^{u_i}_\mathrm{Q}$, in which \cite{WongCuma} 
\begin{align}\label{cumaeq2}
    r^{u_i}_\mathrm{I}=\sum_{j_i\in \mathcal{K}_I}^{}\mathrm{real}\left (r_{j_i}^{(u_i)}  \right ), \hspace{2mm}
r^{u_i}_\mathrm{Q}=\sum_{k_i\in \mathcal{K}_Q}^{}\mathrm{imag}\left (r_{k_i}^{(u_i)}  \right ),
\end{align}
with $r_j^{(u_i)}\stackrel{\Delta}{=}\left [ {\bf{r}}^{u_i} \right ]_j$ denoting the selection of a set of indices $j$ over the entire signal in \eqref{cumaeq1}. From the above definitions, the signal-to-interference-plus-noise ratio ($\mathrm{SINR}$) for UE $u_i$ based on $r^{u_i}_\mathrm{I}$ is expressed as 
\cite{WongCuma}
\begin{align}\label{cumaeq3}
\mathrm{SINR}_{\mathrm{I}}^{u_i}=\frac{P L_{u}\sigma^2_{s} \nu_I^{u_i} }{\delta_i P L_{ u}\sigma^2_{s}\xi^{u_\mathrm{B}}+\frac{\overline{N}\sigma^2_{\eta_{u_i}}}{2}}\stackrel{(a)}\approx\frac{\nu_I^{u_i}}{\delta_i\xi^{u_i}},
\end{align}
where $P$ is the transmit power, $L_{ u}$ are the propagation path-loss associated to user $u$, $\sigma^2_{s}$ is the average symbol power, $\sigma^2_{\eta_{u_i}}$ is the total AWGN power noise at the UE $u_i$, and $\overline{N}=\left | \mathcal{K}_I \right |$. Note that in step $(a)$, the noise power is dropped 
because the system is operating in an interference-limited scenario, i.e., $\mathcal{I}\stackrel{\Delta}{=} U-1 \gg 1$ with $\mathcal{I}$ being the number of interfering users. 
 The parameter $\delta_i\leq1$ is included to allow for the flexibility of incorporating a partial interference cancellation feature\footnote{To perform IC, the BS transmits known orthogonal sequences to UE $u_\mathrm{B}$ for estimating interfering channels using traditional methods.}. For now, let us neglect the role of $\delta_i$, i.e. consider it equals one for the traditional CUMA scenario\footnote{As will see later, this term plays a pivotal role when dealing with CUMA-aided secure communications.}. With these assumptions, that lead to the baseline CUMA scheme in \cite{WongCuma}, and focusing on the reception of the I component (steps to derive the Q component are omitted for conciseness, but are formally equivalent), the signal-to-interference ratio ($\mathrm{SIR}$) can be expressed as $\mathrm{SIR}_{\mathrm{I}}^{u_i}=\tfrac{\nu_I^{u_i}}{\delta_i\xi^{u_i}}$, in which 
\begin{align}\label{cumaeq4}
\nu_I^{u_i}=&\left (  \sum_{k=1}^{\overline{N}} \mathrm{max}\left \{ 0,\mathrm{real}\left ( \left [ {\bf{g}}^{{(u_\mathrm{B},u_i)}}_k \right ]  \right ) \right \} \right )^2 \nonumber \\  \xi^{u_i}=&\sum_{\widetilde{u}_\mathrm{B}=1}^{\mathcal{I}}\sum_{k=1}^{\overline{N}}t_k \mathrm{real}\left ( \left [ {\bf{g}}_{k}^{(\widetilde{u}_\mathrm{B},u_i)} \right ]  \right ),
\end{align}
where $t_k \in \left \{ 0,1 \right \}$ is an i.i.d. Bernoulli random variable with equal probability for modeling the on-off features of a given port $k$. While ${\bf{g}}^{u_\mathrm{B},u_i}_k$ is assumed to be known (estimated) at $u_i$, ${\bf{g}}_{k}^{(\widetilde{u}_\mathrm{B},u_\mathrm{B})}$ does not need to be known a priori at the receiver sides. Inspection of \eqref{cumaeq3} and \eqref{cumaeq4} reveals that the effects of path loss and transmit power are canceled out in the interference-limited regime; hence, the SIR gain \textit{only} comes from the port selection strategy that combines coherently the desired signal and randomly the interference, and from the choice of port density. However, we note that when Bob and Eve are equipped with a CUMA of the same kind, their corresponding SIRs become i.i.d. regardless of the distance and transmit power, and no additional gain can be obtained from a PLS perspective. This conflicts with CUMA's original conception, where CSI-based precoding and power control at the BS are unneeded, as well as estimating the interference channels corresponding to the $\widetilde{u}_\mathrm{B}$ symbols.

To achieve a better understanding of the potential of CUMA scheme for PLS, we evaluate the performance of a modified receiver with partial interference cancellation capabilities; these are captured by \textcolor{black}{introducing} the parameter $\delta_i\in(0,1]$ \textcolor{black}{in the $\mathrm{SIR}_{\mathrm{I}}^{u_i}$} as anticipated previously. Based on the above intuitions, the PDF of $\mathrm{SIR}_{\mathrm{I}}^{u_i}$, denoted by $Z_{I_i}$, after some mathematical manipulations, can be formulated as \cite[eq. (35)]{WongCuma}
\begin{align}\label{cumaeq5}
&f_{Z_{I_i}}(z)=\frac{\left ( \delta_i\sigma_{2,i}^2 \right )^{\frac{1}{4}}\Gamma\left ( \frac{I+1}{2} \right )z^{-\frac{3}{4}}}{\left ( \Gamma\left ( \frac{\mathcal{I}}{2}  \right ) \right )^2 2^{\frac{\mathcal{I}}{2}}\mu_i^{\frac{1}{2}}} \nonumber\\
&\hspace{0.2cm}\times \exp\left(-\frac{\mu_i^2}{4\sigma_{1,i}^2 }\left ( \frac{2\sigma_{1,i}^2+\delta_i\sigma_{2,i}^2z}{\sigma_{1,i}^2+\delta_i\sigma_{2,i}^2z} \right )\right)  \left ( \frac{2}{1+\frac{\delta_i\sigma_{2,i}^2z}{\sigma_{1,i}^2}} \right )^{\frac{2\mathcal{I}+1}{4}} \nonumber \\
&\hspace{0.2cm}\times\mathcal{M}_{-\frac{2\mathcal{I}+1}{4},-\frac{1}{4}}\left ( \frac{\mu_i^2 \delta_i\sigma_{2,i}^2z }{2\sigma_{1,i}^2 \left ( \sigma_{1,i}^2+\delta_i\sigma_{2,i}^2z \right )
} \right ),
\end{align}
where $\mu_i= \tfrac{\overline{N}}{2}\sqrt{\tfrac{\Omega_i}{\pi}}$,  $\Gamma(\cdot)$ is the gamma function, $\Omega_i$ denotes the average channel power, $\mathcal{M}_{a,b}(t)$ is the Whittaker $M$ function
given by~\cite[Sec. 9.220]{Gradshteyn}. Also,~\cite[Eq.~(33)]{WongCuma}
\begin{align}\label{cumaeq12} \sigma_{2,i}^2=\frac{\Omega_i}{4}\left ( \overline{N}+\sum_{m=2}^{\overline{N}}\sum_{k=1}^{m-1}\rho_{k,m} \right ),
\end{align}
and

\begin{align}\label{cumaeq6}
    \sigma_{1,i}^2=\frac{\overline{N}\Omega_i}{4}\left (1-\frac{1}{\pi}  \right )+2\sum_{m=2}^{\overline{N}}\sum_{k=1}^{m-1}\text{cov}_i\left ( X_k,X_m \right ),
\end{align}

\begin{align}\label{cumaeq7}
    \text{cov}_i\left ( X_k,X_m\right )=&\frac{\left ( 1-\rho_{k,m}^2 \right )^{\frac{3}{2}}\Omega_i}{4\pi}-\frac{\Omega_i}{4\pi}+\frac{\rho_{k,m}}{2\sqrt{\pi\Omega_i}}
    \nonumber \\ \times &
    \mathcal{W}\left ( -\sqrt{\frac{2}{1-\rho_{k,m}^2}}\frac{\rho_{k,m}}{\sqrt{\Omega_i}} ,\frac{1}{\Omega_i},\frac{1}{2}\right ),
\end{align}
where 
\begin{align}\label{cumaeq8}
\mathcal{W}(a,b,c)=-\frac{a\Gamma\left ( d_{c} \right )}{\sqrt{2\pi}b^{d_{c}}}{}_2F_1\left ( \frac{1}{2},d_{c};\frac{3}{2};-\frac{a^2}{2b}\right )+\frac{\Gamma\left ( c+1 \right )}{2b^{c+1}},
\end{align}
where $d_{c}=(2c+3)/2$, ${}_2F_1(\cdot,\cdot;\cdot;\cdot)$ is the Gauss hypergeometric function~\cite[Eq.~(9.111)]{Gradshteyn}, and
\begin{align}\label{cumaeq9}
\rho_{k,m}=j_0\left ( 2\pi \sqrt{\left (  \frac{\left ( n_1-n_3 \right )W_1}{N_1-1}\right )^2\hspace{-0.05cm}+\hspace{-0.05cm}\left ( \frac{\left ( n_2-n_4 \right )W_2}{N_2-1} \right )^2} \right ),
\end{align}
is the correlation coefficient, where $j_0(x)=\sin(x)/x$, and the $(n_1,n_2)$-th and the $(n_3,n_4)$-th ports are adequately mapped from the $k$-th and the $m$-th ones as $k=\text{map}(n_1,n_2)$, and $m=\text{map}(n_3,n_4)$, respectively.  In this way, the $(n_1,n_2)$-th port\footnote{Likewise, the $(n_3,n_4)$-th port is calculated from \eqref{eq6cuma} by changing $k$ and $N_1$ for $m$ and $N_2$, respectively.} is computed as 
\begin{align}\label{eq6cuma}
n_1 = \begin{dcases*} N_1, &\hspace{-1mm} if $k\hspace{0.5mm}\text{mod}\hspace{0.5mm}N_1=0$\\ k\hspace{0.5mm}\text{mod}\hspace{0.5mm}N_1, &\hspace{-1mm} otherwise\end{dcases*}\\
n_2 = \begin{dcases*} \left \lfloor \frac{k}{N_1}\right \rfloor,
&\hspace{-1mm} if $n_1=N_1$ \\\left \lfloor \frac{k}{N_1}\right \rfloor
+1, &\hspace{-1mm} o.w., \end{dcases*}
    \end{align}
where $\left \lfloor \cdot \right \rfloor
$ is the floor function. Finally, the PDF of the $\mathrm{SIR}$ of a CUMA UE $u_i$ combines the in-phase and the quadrature component channels, i.e., $\mathrm{SIR}^{u_i}=  \mathrm{SIR}_{\mathrm{I}}^{u_i} +\mathrm{SIR}_{\mathrm{Q}}^{u_i}$. So, mathematically speaking, the PDF of $\mathrm{SIR}^{u_i}$, denoted by $Z_i$, is given by~\cite[Eq.~(38)]{WongCuma}
\begin{align}\label{cumaeq11}
f_{Z_i}(z_i)=\int_{0}^{z_i}f_{Z_{I_i}}(x)f_{Z_{Q_i}}(z_i-x)dx,
\end{align}
in which the random variables ${Z_{I_i}}$ and ${Z_{Q_i}}$ are i.i.d. and each has the PDF given in \eqref{cumaeq5}.
\section{PERFORMANCE ANALYSIS}
In this section, we provide the exact and approximate metrics for the performance evaluation of both the conventional CUMA and the CUMA-aided secure communications case.
\subsection{Conventional CUMA scheme}
Here, we propose an approximate expression for the PDF in \eqref{cumaeq11}. This will be later used to evaluate classical benchmark performance metrics such as the OP and the ER for the conventional CUMA case, for which their exact definitions are given in the sequel.
\subsubsection{Exact Metrics}
\textbf{ER:} From~\cite[Eq.~(43)]{WongCuma}, the ER is defined as the achievable transmission rate per unit bandwidth {for} all i.i.d. UEs $u_\mathrm{B}$. Mathematically speaking, the exact ER is given by 
\begin{align}\label{eq8cuma}
C=U\int_{0}^{\infty}\log_2\left ( 1+\frac{z_\mathrm{B}}{\sigma_{2,\mathrm{B}}^2} \right )f_{Z_\mathrm{B}}(z_\mathrm{B})dz_\mathrm{B} \hspace{2mm}   \text{(bits/channel-use)}.    
\end{align}
\textbf{OP:} This metric is defined as the probability that the channel capacity falls below a predefined threshold rate $\gamma_{th}$. Therefore, from~\cite{WongCuma},
the exact OP of any typical UE $u_\mathrm{B}$ in the conventional CUMA network can be defined as 
\begin{align}
\label{eq10cuma}
P_{\mathrm{out}}& =F_{Z_\mathrm{B}}\left ( z_{th}\right ), 
\end{align}
in which $z_{th}=\left ( 2^{\gamma_{th} } -1\right ) \sigma_{2,\mathrm{B}}^2$, and $F_{Z_\mathrm{B}}(\cdot)$ denotes \ac{CDF} obtained from the PDF defined in~\eqref{cumaeq11}. Here, it is worth highlighting that the exact solutions of ER and OP in~\eqref{eq8cuma} and \eqref{eq10cuma} are expressed in the form of double-fold integrals involving special functions. To circumvent the referred limitation, a simple yet accurate closed-form approximation is proposed for~\eqref{cumaeq11}. This result aims to provide a more straightforward method for calculating the OP and ER expressions.

\subsubsection{Proposed Approximations}
In this section, we begin by approximating the PDF in~\eqref{cumaeq5} by using the novel asymptotic matching method (AoM\footnote{{The main idea behind AoM relies on the fact that parameters of the approximate distribution (e.g., Gamma in this case) are calibrated so as to match its asymptotic behavior to that of the exact distribution of CUMA network. By doing so, the AoM renders a good fit in the left tail and the body of the exact distribution. Recall that the PDF's asymptote around zero governs the high-signal-to-noise-ratio performance of a communication system operating over a channel modeled by that PDF. Conversely, the left tail (non-asymptotic regime) of the PDF plays a minor role, if any. Interested readers can refer to~\cite{Perim} for more detailed information about the AoM approach.}}), as stated in the following proposition.
\begin{proposition}\label{Propos1}
The PDF of $\mathrm{SIR}_{\mathrm{I}}^{u_i}$ for the in-phase channel component given  in~\eqref{cumaeq5} is approximated by
\begin{align}\label{eq11cuma}
f_{Z_{I_i}}(z)&\approx\frac{z^{-\frac{1}{2}}
 }{{\sqrt{\pi \beta_{I_i}} } }\exp\left ( -\tfrac{z}{\beta_{I_i}
} \right ), \hspace{1mm}{z\geq0,}
\end{align}
where
\begin{align}
\label{eq12cuma} 
\beta_{I_i}=\left ( \frac{2\sigma_{1,i}\Gamma\left ( \frac{\mathcal{I}}{2} \right )\mathcal{I}^{-1}}{\sigma_{2,i} \sqrt{\delta_i\pi}\exp\left ( -\frac{ \mu_i^2}{2\sigma_{1,i}^2 } \right ) }\right )^2.
\end{align}
\end{proposition}
\begin{proof}
See Appendix.
\end{proof}
With~\eqref{eq11cuma} at hand, we then illustrate that the gamma distribution\footnote{From a practical viewpoint, this distribution has been widely employed to model the end-to-end channel due to its mathematical tractability, simplifying the derivation of performance metrics in intricate wireless networks.} is a highly accurate approximation for the PDF in~\eqref{cumaeq11}, which is given in the following proposition.
\begin{proposition}\label{Propos2}
The PDF and CDF of the $\mathrm{SIR}^{u_i}$ of the CUMA scheme can be approximated as
\begin{align}\label{eq13cuma}
f_{Z_i}(z_i)\approx &  \frac{1
 }{{\beta_{I_i}}}\exp\left ( -\frac{z_i}{\beta_{I_i}
} \right ), 
\end{align}
\begin{align}\label{eq13cumaCDF}
F_{Z_i}(z_i)\approx1-\Gamma\left (1, \frac{z_i}{\beta_{I_i}
} \right ),
\end{align}
where $\Gamma(\cdot,\cdot)$ is the upper incomplete gamma function~\cite[Eq.~(8.350.2)]{Gradshteyn}.
\end{proposition}
\begin{proof}
See Appendix.
\end{proof}
Next, we obtain closed-form expressions for the OP and the EC for the conventional CUMA scheme {to subsequently explore} the impact of system parameters on the CUMA's performance. Thus, such approximated baseline metrics are given in the following propositions
\begin{proposition}\label{Propos3}
An approximate ER of the conventional CUMA network can be expressed as
\begin{align}\label{eq14cuma}
C\approx \frac{ \textit{U} \exp\left ( \frac{\sigma_{2,\mathrm{B}}^2}{\beta_{I_\mathrm{B}}} \right )\Gamma\left ( 0,  \frac{\sigma_{2,\mathrm{B}}^2}{\beta_{I_\mathrm{B}}}\right )}{\log_e(2)},
\end{align}
where $\log_e(\cdot)$ is the natural logarithm.
\end{proposition}
\begin{proof}
See Appendix.
\end{proof}
\begin{proposition}\label{Propos4}
The OP expression of any typical UEs $u_\mathrm{B}$ of the conventional CUMA network can be approximated by
\begin{align}\label{eq15cuma}
P_{\mathrm{out}}(\gamma_{th})\approx  1-\Gamma\left ( 1,\frac{z_{th}}{\beta_{I_\mathrm{B}}} \right ).
\end{align}
\end{proposition}
\begin{proof}
{$P_{\mathrm{out}}$
is derived by putting~\eqref{eq13cumaCDF} into~\eqref{eq10cuma} with the respective substitutions.}
\end{proof}
\begin{remark}\label{remark1}
Unlike the novel statistic distributions introduced in~\eqref{eq14cuma} and~\eqref{eq15cuma}, as well as the performance metrics in~\eqref{eq8cuma} and~\eqref{eq10cuma} of the CUMA network, where the expressions were derived in single- and double-integral form, respectively; our approximations are simple, highly accurate, and do not {require solving} any involved integrals nor incur special functions.
\end{remark}


\subsection{CUMA-aided secure communications scheme}
\subsubsection{Exact Metric}
\textbf{SOP:} This is a useful metric for quantifying leakage information in practical setups where passive eavesdropping is considered, i.e., the \ac{CSI} of the eavesdropper is not available at Alice. In this scenario, the UE $u_\mathrm{E}$ remains silent during all transmission and only intercepts a message of a selected UE $u_\mathrm{B}$ without communicating with other nodes of the CUMA network. Based on this, the secrecy rate 
\textcolor{black}{of a typical legitimate UE in the CUMA network}
 is defined as $C_\mathrm{S}=\mathrm{max}[0,C_{u_\mathrm{B}}-C_{u_\mathrm{E}}]$, where $C_{u_\mathrm{B}}=\log_2\left ( 1+z_\mathrm{B} \right )$ and $C_{u_\mathrm{E}}=\log_2\left ( 1+z_\mathrm{E} \right )$ are the rates of the main and eavesdropper channels, respectively.
The $\mathrm{SOP}$ is formulated as the probability that the secrecy rate 
drops below a target secrecy rate $R_\mathrm{S}$. Hence, since all UEs are i.i.d., \textcolor{black}{ the $\mathrm{SOP}$ of any typical legitimate UE $u_\mathrm{B}$ in the CUMA network is given by}~\cite{Barros}
\begin{align}\label{cumaeq13}
\mathrm{SOP}&=\Pr\left \{  C_\mathrm{S} < R_{\mathrm{S}}  \right \}\nonumber \\
 & =\int_{0}^{\infty}F_{Z_\mathrm{B}}\left ( \tau\left ( 1+z_\mathrm{E}\right ) -1 
\right )f_{Z_\mathrm{E}}(z_\mathrm{E})dz_\mathrm{E},
\end{align}
where $\text{Pr}\left \{ \cdot \right \}$ denotes probability, $F_{Z_\mathrm{B}}(\cdot)$ is obtained from the PDF in \eqref{cumaeq11}, and  $\tau\buildrel \Delta \over  = 2^{R_{\mathrm{S}}}$. From~\eqref{cumaeq13}, a lower bound of the $\mathrm{SOP}$, denoted as $\mathrm{SOP}_{\mathrm{L}}$, can be obtained as
\begin{align}\label{lower}
\mathrm{SOP}_{\mathrm{L}}=\Pr\left \{ z_\mathrm{B}< \tau  z_\mathrm{E} 
\right \} \leq \mathrm{SOP}.    
\end{align}

\subsubsection{Proposed SOP Approximation}
With the channel statistics at hand, i.e., $f_{Z_\mathrm{E}}(z_\mathrm{E})$ and $ F_{Z_\mathrm{B}}(z_\mathrm{B})$, which are given in closed-form fashion in \eqref{eq13cuma} and \eqref{eq13cumaCDF}, respectively. In the following proposition, we obtain a closed-form expression for the SOP lower bound of the CUMA-assisted PLS system to examine the effect of system parameters on the secrecy performance.
\begin{proposition}\label{Propos5}
The $\mathrm{SOP}_{\mathrm{L}}$ expression of the the CUMA-aided secure communications case can be approximated by
\begin{align}\label{cumaeq17}
\mathrm{SOP}_{\mathrm{L}}\approx 1-\frac{\beta_{I_\mathrm{B}}}{\tau \beta_{I_\mathrm{E}} +\beta_{I_\mathrm{B}}}.
\end{align}
\end{proposition}
\begin{proof}
See Appendix.
\end{proof}
\begin{remark}\label{remark1}
\textcolor{black}{
Notice that unlike the exact SOP in~\eqref{cumaeq13}, which is given in the double-integral form, our lower bound approximation of the SOP in~\eqref{cumaeq17} is simple and does not incur any sophisticated special functions either.}
\end{remark}

\section{NUMERICAL RESULTS AND DISCUSSIONS} \label{sect:numericals}
\begin{figure*}[h!]
    \centering
\psfrag{A}[Bc][Bc][0.5]{$\mathrm{\textit{N}=36}$}
\psfrag{B}[Bc][Bc][0.5]{$\mathrm{\textit{N}=144}$}
    \subfigure[$f=6$ GHz] {\includegraphics[width=0.32\textwidth]{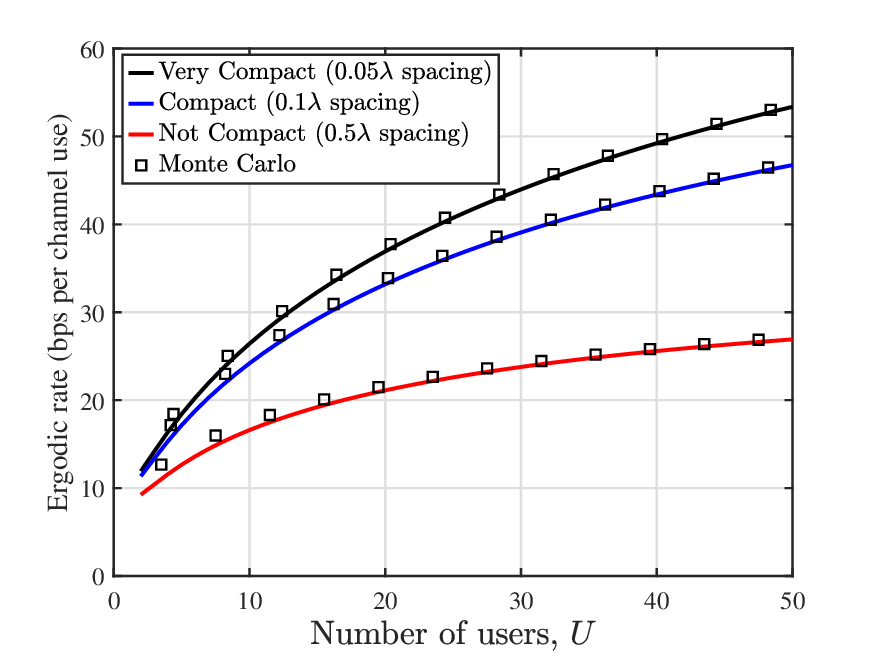}} 
\psfrag{A}[Bc][Bc][0.5]{$\mathrm{\textit{N}=36}$}
\psfrag{B}[Bc][Bc][0.5]{$\mathrm{\textit{N}=100}$}
\psfrag{C}[Bc][Bc][0.5]{$\mathrm{\textit{N}=256}$}
    \subfigure[$f=26$ GHz]{\includegraphics[width=0.32\textwidth]{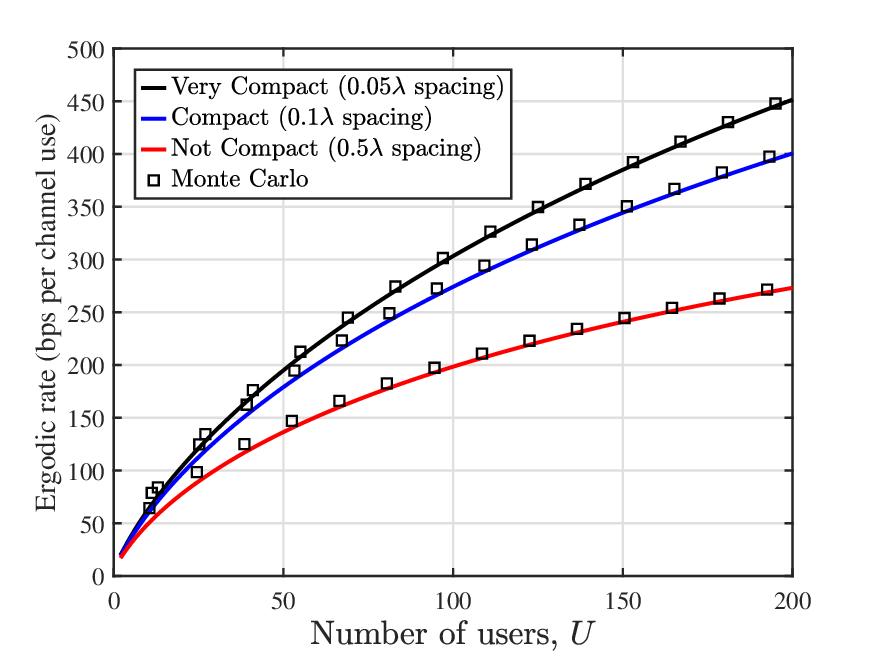}}
\psfrag{A}[Bc][Bc][0.5]{$\mathrm{\lambda/12}$}
\psfrag{B}[Bc][Bc][0.5]{$\mathrm{\lambda/8}$}
\psfrag{C}[Bc][Bc][0.5]{$\mathrm{\lambda/4}$}
\psfrag{D}[Bc][Bc][0.5]{$\mathrm{i.i.d. \hspace{1mm}}\mathrm{fading}$}
    \subfigure[$f=40$ GHz]{\includegraphics[width=0.32\textwidth]{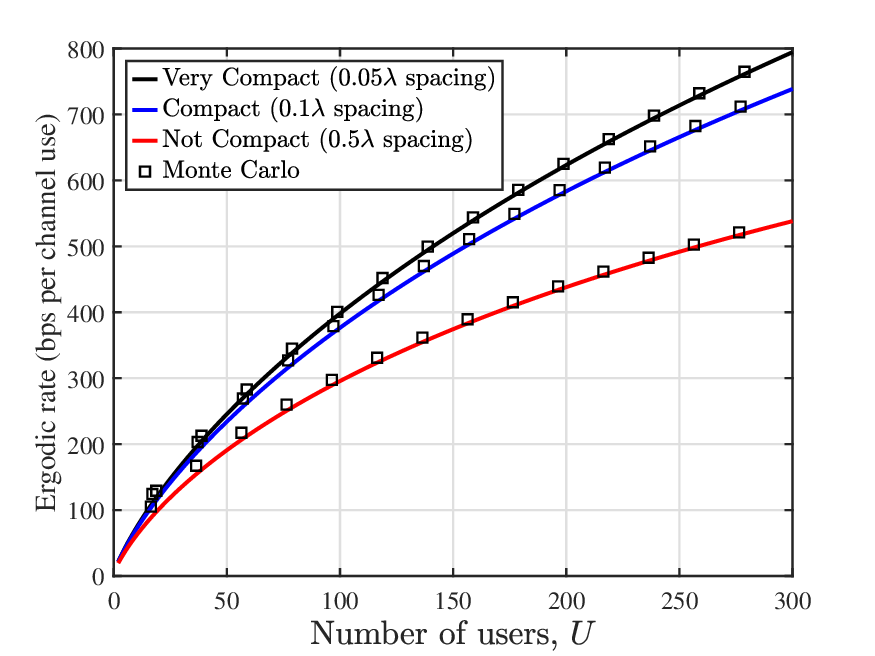}}
    \caption{ER vs. the number of users, $U$, for different operating frequencies under rich scattering. Markers denote Monte Carlo Simulations, whereas the solid line represents the proposed approximation given in~\eqref{eq14cuma}.}
    \label{fig:foobar}
\end{figure*}

In this section, we evaluate the effect of the system parameters (e.g., port density, number of users, and antenna size) on the two CUMA-based scenarios: $i)$ the conventional CUMA and $ii)$ the CUMA-aided secure communications scheme, as well as the goodness-of-fit of the proposed approximations. For this purpose, the settings assumed for the port density and frequencies are shown in Table~\ref{table1}, where the cases under analysis are:  Case I: a non-compact (NC) case with a minimum spacing of $0.5\lambda$, Case II:  a compact (C) case with a minimum spacing of $0.1\lambda$, and Case III: a very compact (VC) case with a minimum spacing of $0.05\lambda$. Besides, the size of the 2D FAS at each UE $u_i$ for $i \in \left \{ \mathrm{B}, \mathrm{E}  \right \}$ is assumed to be $15$ cm $\times$ 8 cm, which is a typical mobile phone size\footnote{It is worth mentioning that the number of ports, i.e., $N_1$ and $N_2$, are chosen so that they do not exceed the typical cell phone size when distributed over $W$.}. Specifically, by matching $( W_1\lambda, W_2 \lambda)=(15,8)$, the $W_j$ values   with $j \in \{1,2 \}$ can be estimated. With such values, the port density can be calculated as $N_1\approx\tfrac{W_1}{x}+1$, where $x \in \{0.05, 0.1, 0.5  \}$ is the minimum space depending on the CUMA operating case. Likewise, $N_2$ is calculated as $N_1$ with the respective parameter substitutions. Unless remarked otherwise, for all plots, we consider that $\Omega_i=\Omega=1$, and the correlation model for the 2D FAS at the UEs $u_i$ is formulated from~\eqref{cumaeq9}. Also, we provide illustrative numerical results along with Monte Carlo simulations to verify the proposed analytical derivations.
Next, each of the proposed CUMA scenarios is analyzed.

\begin{table}[t!]
	\scriptsize 
    	\caption{CUMA network parameter settings.} 
   \centering
	\begin{tabular}{cccc}
		\toprule
		\multicolumn{1}{c}{\multirow{1}{*}{}} & \multicolumn{2}{c}{\textbf{\hspace{6mm}Compactness or port density, $N_1\times N_2$ }}\\
		\cmidrule(lr){2-4}
		  \textit{ \textbf{f}}\textbf{(GHz)}  &  \textbf{Case I (NC)} &  \textbf{Case II (C)} & \textbf{Case III (VC)} \\
		\cmidrule(lr){1-4}
		6  & 7$\times$4    &31$\times$4  &61$\times$4 \\
		\cmidrule(lr){2-4}
		26 & 27$\times$14  & 131$\times$14  &261$\times$14 \\
		\cmidrule(lr){2-4}
		40 & 41$\times$22 &201$\times$22  &401$\times$22  \\
     	\cmidrule(lr){1-4}
	\end{tabular}\label{table1}
	\vspace{-5mm}
\end{table}

\subsection{Conventional CUMA Analysis}
In this section, we evaluate the effect of the system parameters on the OP and ER behavior in the conventional CUMA network under rich scattering.\footnote{{As originally conceived, the CUMA scheme can operate in two different channel conditions, namely, i) a finite scatterer channel model, and ii) a rich scattering channel model (characterized by a Rayleigh fading). The most beneficial channel condition for CUMA is a rich scattering because this setup no longer needs a precoding scheme at BS to inject sufficient channel differentiation between the UEs for FAS to work well.}}

In Figs.~\ref{fig:foobar}a-\ref{fig:foobar}c, the ER is illustrated as a function of the number of users $U$ by varying the compactness, i.e., the port density at the FAS-enabled UE for different frequency operations. Note that our proposed approximations perfectly match the Monte Carlo simulations in all plots of such figures. Now, in Fig.~\ref{fig:foobar}a, it can be observed that packing the antenna ports more densely is favorable to the ER performance, as expected. For instance, note that Case III (VC), where the adjacent antenna ports are so tiny (i.e., the ports are $0.05\lambda$ apart), {there is better} ER performance. Conversely, 
when the antenna port separation is pushed to $0.5\lambda$, as in Case I (NC), a lower ER behavior appears. Therefore, the port density at the FAS-enabled UE arises as an essential factor for boosting the performance of CUMA at $6$ GHz frequency operation. The previously described CUMA behavior applies to Figs.~\ref{fig:foobar}b and~\ref{fig:foobar}c as well. Likewise, from Figs.~\ref{fig:foobar}a-\ref{fig:foobar}c, it is evident that the ER in the CUMA network is proportional to $U$; i.e., as the number of users $U$ increases, the ER performance also improves monotonically. 
Hence, this feature turns CUMA into an attractive technique for simple massive connectivity without the need for sophisticated methods {to mitigate} interference at the UE or precoding at the BS. Finally, from all traces in these figures, it is shown that when the number of ports increases {due to} higher frequencies (see Table~\ref{table1}), the ER performance is {remarkably enhanced}. 
\begin{figure}[t!]
\centering
\psfrag{A}[Bc][Bc][0.7]{$\mathrm{\textit{n}=49}$}
\psfrag{B}[Bc][Bc][0.7]{$\mathrm{\textit{n}=100}$}
\psfrag{C}[Bc][Bc][0.7]{$\mathrm{\textit{n}=196}$}
\includegraphics[width=80mm]{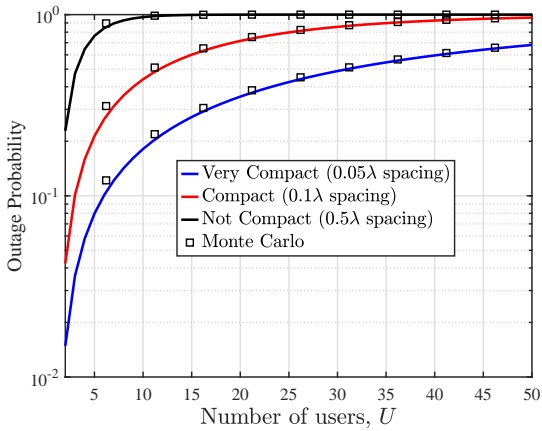}
\caption{{\ac{OP} of any typical UE of the CUMA network by varying $U$ with a fixed frequency operation of $f=6$ GHz under rich scattering. } Markers denote Monte Carlo Simulations, whereas the solid line represents the proposed approximation given in~\eqref{eq15cuma}.}
\label{figOPCUMA}
\end{figure}
\begin{figure}[t!]
\centering
\psfrag{A}[Bc][Bc][0.7]{$\mathrm{\textit{n}=49}$}
\psfrag{B}[Bc][Bc][0.7]{$\mathrm{\textit{n}=100}$}
\psfrag{C}[Bc][Bc][0.7]{$\mathrm{\textit{n}=196}$}
 \includegraphics[width=80mm]{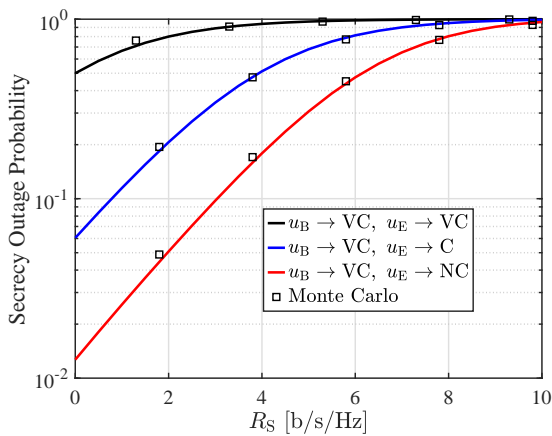}
\caption{SOP vs. the threshold rate, $R_{\rm S}$, by assuming $f=6$ GHz, $\delta_i=1$, and $U=20$.  Markers denote Monte Carlo Simulations, whereas the solid line represents the proposed approximation given in~\eqref{cumaeq17}. }
\label{fig1SOP}
\end{figure}

In Fig.~\ref{figOPCUMA}, we exemplify the OP vs. the number of users $U$ by considering a fixed frequency operation of $6$ GHz and the configuration parameters given in cases I, II, and III, as described in Table~\ref{table1}. In this scenario, we explore the impact of having different port densities on the OP performance of any typical UE of the CUMA network. From all curves, it can be observed that increasing the number of UEs (i.e., interference) leads to poor OP performance, as expected. Also, we see that the OP performance depends on the different levels of compactness; for instance, the denser port configuration (i.e., Case III) performs better than Cases I and II (less compact). This OP behavior could improve notably by increasing the operating frequency.
\subsection{CUMA-aided Secure Communication Analysis}
In this section, we explore the impact of the system parameters on the $\mathrm{SOP}$ of proposed CUMA schemes, also analyzing the effect of an IC capability for the legitimate user $u_\mathrm{B}$.

In Figs.~\ref{fig1SOP} and~\ref{fig2SOP}, we show the results of the SOP assuming the CUMA setup as it was initially conceived, i.e., no precoding optimization at the BS, nor IC at the UEs. Specifically, in Fig.~\ref{fig1SOP}, we present the $\mathrm{SOP}$ as a function of the threshold rate $R_{\rm S}$ for $\delta_i=1$ and different port density configurations, i.e., compactness at UE $u_i$ operating at $6$ GHz frequency operation. Here, it can be observed that when the antenna ports are more densely packed at the legitimate UE, i.e., $u_\mathrm{B}\rightarrow \mathrm{VC}$ (the ports are $0.05\lambda$ apart), and packing the antenna ports less densely at the eavesdropper 
UE, i.e., $u_\mathrm{E}\rightarrow \mathrm{NC}$ (the port separation is pushed to $0.5\lambda$), a reasonably good secrecy performance is attained. Otherwise, when both $u_\mathrm{B}$ and $u_\mathrm{E}$ have the same compactness capacities, i.e., $u_\mathrm{B}\rightarrow \mathrm{VC}$ and $u_\mathrm{E}\rightarrow \mathrm{VC}$, secure communication becomes unfeasible.  
Finally, note that (as in the other figures) the proposed approximations are in excellent agreement with the Monte Carlo simulations.
\begin{figure}[t!]
\centering
\psfrag{A}[Bc][Bc][0.7]{$\mathrm{\textit{n}=49}$}
\psfrag{B}[Bc][Bc][0.7]{$\mathrm{\textit{n}=100}$}
\psfrag{C}[Bc][Bc][0.7]{$\mathrm{\textit{n}=196}$}
 \includegraphics[width=80mm]{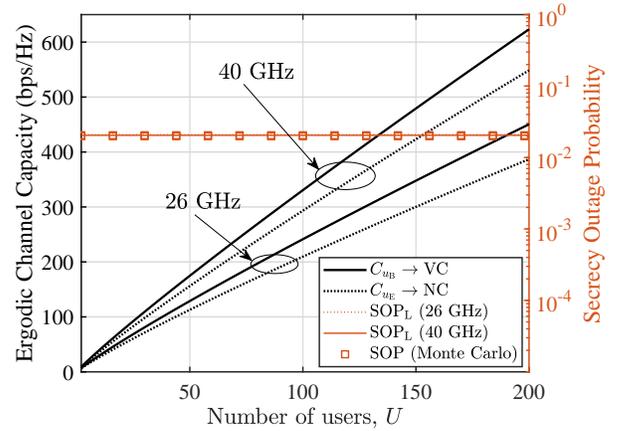}
\caption{SOP vs. the number of users, $U$, with $f=26/40$ GHz, $\delta_i=1$, and \textcolor{black}{$R_\mathrm{S}=1$}. Markers denote Monte Carlo Simulations, whereas the solid line represents the proposed approximation given in~\eqref{cumaeq17}.}
\label{fig2SOP}
\end{figure}
In Fig.~\ref{fig2SOP}, we show the $\mathrm{SOP}$ vs. $U$ for different frequency operations in the absence of IC, assuming $u_\mathrm{B}\rightarrow\mathrm{VC}$, $u_\mathrm{E}\rightarrow\mathrm{NC}$. The ergodic channel rates for the individual Bob and Eve's paths are included for reference purposes. From the $\mathrm{SOP}$ plots, it is clear that increasing the operating frequencies does not favor the secrecy performance. In fact, the SOP behavior at these frequencies is practically identical; the main reason for this event is explained as follows. Note that the individual rates for the main and eavesdropper channels, denoted by $C_{u_\mathrm{B}}$ and $C_{u_\mathrm{E}}$, respectively, increase when going from $26$ GHz to $40$ GHz, as expected. However, the capacity gain obtained from the difference between $C_{u_\mathrm{B}}$ and $C_{u_\mathrm{E}}$ is barely affected regardless of the operating frequency. Hence, the secrecy performance when $u_\mathrm{B}$ and $u_\mathrm{E}$ 
have the same port density remains constant, and unaffected by changes in the operation frequency when the FA physical size is not changed. We also observe an irreducible SOP floor as the number of users grows, since for larger $U$ the interference both at Bob and Eve increases proportionally, when the new UEs that join the CUMA network have the same port density.

\begin{figure}[t!]
\centering
\psfrag{A}[Bc][Bc][0.7]{$\mathrm{\textit{n}=49}$}
\psfrag{B}[Bc][Bc][0.7]{$\mathrm{\textit{n}=100}$}
\psfrag{C}[Bc][Bc][0.7]{$\mathrm{\textit{n}=196}$}
 \includegraphics[width=80mm]{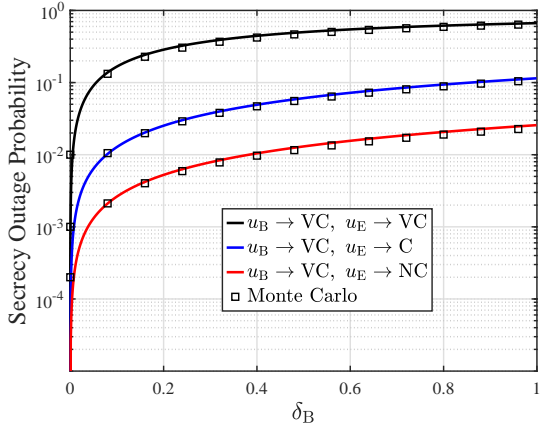}
\caption{SOP vs. the interference cancellation factor, $\delta_{\rm B}$, for different port density configurations. Also, $f=6$ GHz,  $\delta_\mathrm{E}=1$, \textcolor{black}{ $U=20$, and $R_\mathrm{S}=1$}. Markers denote Monte Carlo Simulations, whereas the solid line represents the proposed approximation given in~\eqref{cumaeq17}.}
\label{fig3SOP}
\end{figure}

 Given the limited capability of the baseline CUMA scheme to provide a satisfactory secure performance, we will now analyze the case on which partial IC features are allowed at the receiver side. This is captured by the $\delta_i$ parameter that ranges between $0$ and $1$, implying that as $\delta_i$ is reduced the IC becomes increasingly better. We make the practical assumption that interference mitigation at the legitimate UE $u_\mathrm{B}$ is achieved by using a precoding sequence at the BS to transmit UE $u_\mathrm{B}$'s signals emulating an artificial noise injection technique, for which the eavesdropper is unaware of and hence $u_\mathrm{E}$ cannot cancel it. Based on these premises, in Fig.~\ref{fig3SOP}, we examine to what extent implementing an IC capability at Bob's side can be beneficial for secrecy performance. For this, the $\mathrm{SOP}$ vs. $\delta_\mathrm{B}$ is evaluated for different port density settings at both UEs $u_\mathrm{B}$ and $u_\mathrm{E}$, and $f=6$ GHz. Here, unlike previous figures, partial-IC is performed at UE $u_\mathrm{B}$ whereas the UE $u_\mathrm{E}$ cannot mitigate interference, so $ \delta_\mathrm{E}=1$. It becomes clear that $\mathrm{SOP}$ performance enhances as $\delta_\mathrm{B}$ decreases, which reveals that to achieve suitable secrecy levels when both $u_\mathrm{B}$ and $u_\mathrm{E}$ are equipped with FAS, it is necessary to combine the CUMA technology with some additional signal processing capabilities (even simple) at the target UE $u_\mathrm{B}$. Note that for very low $\delta_\mathrm{B}$, the interference-limited assumption is not be met, and secrecy performance will be limited by noise under these circumstances.

\textcolor{black}{In Fig.~\ref{fig4SOP}, we explore the secrecy performance by varying the IC capability at UE $u_\mathrm{B}$, i.e., $\delta_\mathrm{B}$ and assuming different combinations of compactness capacities at $u_\mathrm{B}$ and $u_\mathrm{E}$ with a fixed $f = 6$ GHz, and $\delta_\mathrm{E}=1$. Specifically, for $u_\mathrm{E}\rightarrow\mathrm{NC}$, we set $(N_1\times N_2)=(7\times4)$, so $N=28$ ports in this case. Conversely, for $u_\mathrm{B}\rightarrow\mathrm{VC}$, we consider that $(N_1\times N_2)$ goes from $(61\times4)$ to $(61\times 33)$, i.e., $N$ ranging from $244$ to $2013$ for this scenario. It can be observed that the combination of high IC capability (e.g., $\delta_\mathrm{B}=0.1$) with an increased number of ports in the VC scheme at UE $u_\mathrm{B}$ is favorable from a PLS perspective. This confirms that secrecy performance is enhanced remarkably thanks to the imperfect IC capability and extreme VC configuration at the legitimate UE. Finally, note that in all curves, our analytical expressions, for OP, ER and lower bound of the SOP, perfectly match with Monte Carlo simulations.}

\begin{figure}[t!]
\centering
\psfrag{A}[Bc][Bc][0.7]{$\mathrm{\textit{n}=49}$}
\psfrag{B}[Bc][Bc][0.7]{$\mathrm{\textit{n}=100}$}
\psfrag{C}[Bc][Bc][0.7]{$\mathrm{\textit{n}=196}$}
 \includegraphics[width=80mm]{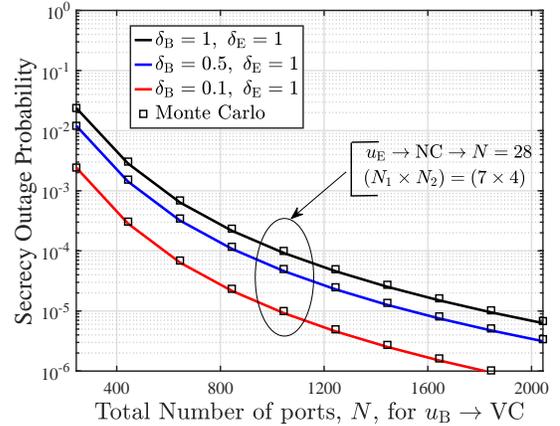}
\caption{SOP by varying the total number of ports, $N$, at UE $u_\mathrm{B}\rightarrow\mathrm{VC}$ with $f=6$ GHz, $R_\mathrm{S}=1$, $U=20$ and $\delta_\mathrm{E}=1$. Markers denote Monte Carlo Simulations, whereas the solid line represents the proposed approximation given in~\eqref{cumaeq17}.}
\label{fig4SOP}
\end{figure}

\section{CONCLUSION}
In this work, we developed a set of mathematical tools to enable the performance analysis of the novel open-loop massive connectivity scheme called CUMA. Employing asymptotic techniques, we derived a closed-form approximation for the PDF of the SIR. {Using this result, the ER and OP metrics were evaluated straightforwardly, avoiding the need for double-fold integrals.} The results illustrated that CUMA with a good level of compactness (i.e.,
$0.05\lambda$ spacing) provided remarkable gains in terms of the ER performance over the non-compact case.  
Then, the secrecy performance of CUMA was evaluated for the first time in the literature. The PLS performance in a CUMA network was investigated, combined with the role of imperfect interference mitigation techniques at the legitimate UE. Analytical approximate results for the SOP are obtained, which are much more amenable than those obtained using the original formulation in \cite{WongCuma}. Results reveal that one of CUMA's strongest features, i.e., the lack of side-processing mechanisms, limits its achievable secrecy performance, as Bob and Eve SIRs have the same statistics when the same FA scheme is used on both sides. We observe that partially relaxing the original constraints of the CUMA by incorporating some simple IC capability can be beneficial for PLS. Hence, the trade-off between complexity and secure performance highlights CUMA's potential for new security-constrained use cases.

\appendices

\label{appendix1} 

\centerline{
\textbf{Proof of Proposition~\ref{Propos1}}}
As a first step, a suitable approximation for the PDF given in~\eqref{cumaeq5} is needed. For this purpose, we begin by expressing the Whittaker $M$ function in a more tractable mathematical fashion. Hence, with the aid of~\cite[Eq.~(9.220.2)]{Gradshteyn}, it follows that $\mathcal{M}_{a,b}(t)=t^{b+\frac{1}{2}}\exp\left (-\tfrac{t}{2}  \right ) {}_1F_1\left (b-a+\frac{1}{2},2b+1;t  \right )$, where ${}_1 F_1\left(\cdot,\cdot;\cdot\right)$ denotes the confluent hypergeometric function~\cite[Eq.~(9.210.1)]{Gradshteyn}. Next,  by applying the asymptotic relationships:\footnote{The asymptotic behavior around zero of the $\exp(\cdot)$ and $I_{(\cdot)}(\cdot)$ functions is governed by the first term of the Maclaurin series representation of such  functions.}:
$i)$ $\exp(x)\simeq 1$, and $ii)$ $ {}_1F_1\left (a,b;t  \right )\simeq 1+\tfrac{at}{b}$, and after some manipulations, the asymptotic behavior of~\eqref{cumaeq5} in the form $f_{Z_{I_i}}(z)\simeq a_0 z^{b_0}$ can be expressed by 

\begin{align}\label{eq1AP}
    f_{Z_{I_i}}(z)&\simeq\underset{a_0}{\underbrace{\frac{\exp\left ( -\frac{ \mu_i^2}{2\sigma_{1,i}^2 } \right ) \mathcal{I}\sigma_{2,i} \left ( \delta_i \right )^\frac{1}{2}}{2\sigma_{1,i}\Gamma\left ( \tfrac{\mathcal{I}}{2} \right )}}}z^{  \overset{b_0}{\overbrace{-1/2 }} },
\end{align}
where $\mu_i= \tfrac{\overline{N}}{2}\sqrt{\tfrac{\Omega_i}{\pi}}$, and $a_0$ and $b_0$ are the linear and the angular coefficients that capture the asymptotic behavior of the exact PDF in~\eqref{cumaeq5}.
Then, with that result,~\eqref{cumaeq5} is approximated by using a simple Gamma distribution whose PDF and asymptotic PDF are given, respectively, by
\begin{multicols}{2}
  \begin{align}\label{eq2AP}
    \widetilde{f}_{X}(x)=\frac{x^{\alpha_{I_i}-1}e^{-\frac{x}{\beta_{I_i}}}}{\Gamma(\alpha_{I_i})\beta_{I_i}^{\alpha_{I_i}}}
  \end{align}
\begin{align}\label{eq3AP}
     \hspace{-2mm}\widetilde{f}_{X}^{\rm asy}(x)\simeq\underset{\widetilde{a}_0}{\underbrace{\tfrac{1}{\Gamma(\alpha_{I_i})\beta_{I_i}^{\alpha_{I_i}}}}}x^{\overset{\widetilde{b}_0}{\overbrace{\alpha_{I_i}-1}}},
  \end{align}
\end{multicols}
By matching the asymptotic behavior of the exact PDF given in~\eqref{eq1AP} with the approximate one in~\eqref{eq3AP}, i.e., $a_0=\widetilde{a}_0$ and $b_0=\widetilde{b}_0$, the scaling parameters of the approximate distribution in~\eqref{eq2AP} are found as
$\alpha_{I_i}=1/2$, and $\beta_{I_i}=\left ( \tfrac{1}{a_0 \sqrt{\pi}} \right )^{2}$.  Next, by substituting~\cite[Eq.~(34)]{WongCuma} and~\eqref{eq2AP} into~\cite[Eq.~(70)]{WongCuma}, it follows that
\begin{align}\label{eq4AP}
f_{Z_{I_i}}(z)&\approx \tfrac{2^{-\frac{\mathcal{I}}{2}}(\sigma_{2,i}^2)^{\frac{\mathcal{I}}{2}-1}z^{-\frac{1}{2} } }{\delta_i^{\frac{\mathcal{I}}{2}-1}\Gamma\left ( \frac{I}{2}\right )\Gamma\left ( \frac{1}{2}\right )\beta_{I_i}^{\frac{1}{2}}}   \underset{I_1}{\underbrace{\int_{0}^{\infty} x^{-\frac{3-\mathcal{I}}{2}}e^{-\tfrac{2\delta_izx+\sigma_{2,i}^2\beta_{I_i} x}{2\delta_i\beta_{I_i}}}dx}}.
\end{align}
Using~\cite[Eq.~(3.351.3)]{Gradshteyn}, $I_{1}$ can be solved in closed-form. Then, performing some manipulations in~\eqref{eq4AP},~\eqref{cumaeq5} can be approximated as stated in~\eqref{eq11cuma}, which completes the proof.\\

\centerline{
\textbf{Proof of Proposition~\ref{Propos2}}}
With~\eqref{eq11cuma} at hand, our goal is to approximate the PDF and CDF of a UE $u_i$ of the CUMA network used in the integral defined in~\eqref{cumaeq11} by using again the Gamma distribution
via the asymptotic matching technique. Hence, the
approximate PDF for a UE $u_i$ of the CUMA scheme can be expressed as
\begin{align}\label{eq5AP}
  \widetilde{f}_{Z_i}(z_i)=\frac{1}{\underset{\widetilde{c}_0}{\underbrace{\Gamma\left ( \alpha_{i} \right ){\beta_{i}}^{ \alpha_{i} } }}}z_i^{\overset{\widetilde{d}_0}{\overbrace{\alpha_{i}-1}}}\exp\left ( -\frac{z_i}{\beta_{i}
} \right ).
\end{align}
To apply the asymptotic matching method, we need to find the asymptotic behavior of~\eqref{cumaeq11}. To this end, by appropriately substituting the shape parameters of the distributions $Z_{I_i}$ and $Z_{Q_i}$ given in~\eqref{eq11cuma}
\footnote{
Since $Z_{I_i}$ and $Z_{Q_i}$ are i.i.d., they have the same distribution. Therefore, the fitting parameters of $Z_{I_i}$ and $Z_{Q_i}$ are given by $\left \{ \alpha_{I_i}^1,\beta_{I_i}^1 \right \}$ and $\left \{ \alpha_{I_i}^2,\beta_{I_i}^2\right \}$, respectively.}
into~\cite[Eq.~(23)]{Perim} and after some manipulations, the linear and angular coefficients that govern the asymptote of~\eqref{cumaeq11}, i.e., $f_{Z_i}(z_i)\simeq c_0 z_i^{d_0}$, are given 
\begin{align}\label{eq6AP}
 c_0 = \frac{\prod_{j=1}^{2}\left [ \tfrac{1}{{\beta_{I_i}^{j}}^{\alpha_{I_i}^{j}} }  \right ]}{\Gamma \left( \sum_{j=1}^{2}\alpha_{I_i}^{j} \right )} , \hspace{2mm}  d_0  = 1+\sum_{j=1}^{2}\alpha_{I_i}^{j}.
\end{align}
Then, by matching $c_0=\widetilde{c}_0$ and $d_0=\widetilde{d}_0$, the fitting parameters are found as $\alpha_i=1$ and $\beta_i=\beta_{I_i}$, which substituted in~\eqref{eq5AP} allow to obtain~\eqref{eq13cuma} as an approximation of~\eqref{eq11cuma}. Finally, the CDF in~\eqref{eq13cumaCDF} of UE $u_i$, i.e.,  $F_{Z_i}(z_i)$ is obtained directly from~\eqref{eq13cuma}. This completes the proof.\\

\centerline{
\textbf{Proof of Proposition~\ref{Propos3}}}
To obtain a closed-form expression of the ER, we replace~\eqref{eq13cuma} with the respective substitutions  into~\eqref{eq8cuma}, this get
\begin{align}\label{eq8AP}
C\approx \frac{U
 }{{\beta_{I_\mathrm{B}}}}\underset{I_2}{\underbrace{\int_{0}^{\infty}\log_2\left ( 1+\frac{z_\mathrm{B}}{\sigma_{2,\mathrm{B}}^2} \right )
\exp\left ( -\frac{z_\mathrm{B}}{\beta_{I_\mathrm{B}}
} \right )dz_\mathrm{B}.}}
\end{align}
The resulting integral $I_2$ can be evaluated in closed form with the aid of~\cite[Eq.~(4.337.2)]{Gradshteyn}; hence,~\eqref{eq8cuma} can be approximated as stated in~\eqref{eq14cuma}, which completes the proof.\\

\centerline{
\textbf{Proof of Proposition~\ref{Propos5}}}
By plugging~\eqref{eq13cuma} and~\eqref{eq13cumaCDF} with the respective parameter substitutions into \eqref{lower}, it follows that
\begin{align}\label{eq9AP}
  \mathrm{SOP}_{\mathrm{L}}\approx  & \underset{I_3}{\underbrace{\int_{0}^{\infty} \frac{1
 }{{\beta_{I_\mathrm{E}}}}\exp\left ( -\frac{z_\mathrm{E}}{\beta_{I_\mathrm{E}}
} \right )dz_\mathrm{E}}}
\nonumber \\ &-\frac{1
 }{{\beta_{I_\mathrm{E}}}}\underset{I_4}{\underbrace{\int_{0}^{\infty}\exp\left ( -\frac{z_\mathrm{E}}{\beta_{I_\mathrm{E}}
} \right )\Gamma\left (1,  \frac{\tau z_\mathrm{E}}{\beta_{I_\mathrm{B}}
} \right )dz_\mathrm{E}}}.
\end{align}
Here, $I_3=1$ since the definite integral of Exponential \ac{PDF} over the entire space is always equal to one. Then, by solving the corresponding integral in $I_4$ with the help of~\cite[Eq.~(3.310)]{Gradshteyn}, we have that
\begin{align}\label{eq10AP}
 I_4=\frac{{\beta_{I_\mathrm{B}}}{\beta_{I_\mathrm{E}}}}{{\beta_{I_\mathrm{B}}}+\tau {\beta_{I_\mathrm{E}}}}.
\end{align}
 Finally, by plugging~\eqref{eq10AP} into~\eqref{eq9AP}, \eqref{cumaeq17} is attained in closed-form formulation. This completes the proof.

\bibliographystyle{ieeetr}
\bibliography{bibfile}

\begin{thebibliography}{10}

\bibitem{Bruno}
B.~Clerckx, Y.~Mao, E.~A. Jorswieck, J.~Yuan, D.~J. Love, E.~Erkip, and D.~Niyato, ``A primer on rate-splitting multiple access: Tutorial, myths, and frequently asked questions,'' {\em IEEE J. Sel. Areas Commun.}, vol.~41, no.~5, pp.~1265--1308, 2023.

\bibitem{nomaa}
H.~R.~C. Mora, N.~V.~O. Garzón, F.~D.~A. García, J.~D.~V. Sánchez, and O.~L.~A. López, ``Secure transmission for uplink scma systems over $\kappa$-$\mu$ fading channels,'' {\em IEEE Trans. Veh. Technol.}, vol.~73, no.~1, pp.~754--770, 2024.

\bibitem{TongNewPAra}
K.-K. Wong, K.-F. Tong, and C.-B. Chae, ``Fluid antenna system—part iii: A new paradigm of distributed artificial scattering surfaces for massive connectivity,'' {\em IEEE Commun. Lett.}, vol.~27, no.~8, pp.~1929--1933, 2023.

\bibitem{Ref2}
K.-K. Wong, A.~Shojaeifard, K.-F. Tong, and Y.~Zhang, ``Fluid antenna systems,'' {\em IEEE Trans. Wireless Commun.}, vol.~20, no.~3, pp.~1950--1962, 2021.

\bibitem{ByoungRO}
K.-K. Wong and et~al, ``Fluid antenna system—part ii: Research opportunities,'' {\em IEEE Commun. Lett.}, vol.~27, no.~8, pp.~1924--1928, 2023.

\bibitem{AccessFAS}
K.-K. Wong and K.-F. Tong, ``Fluid antenna multiple access,'' {\em IEEE Trans. Wireless Commun.}, vol.~21, no.~7, pp.~4801--4815, 2022.

\bibitem{WongfFAMA}
K.-K. Wong, K.-F. Tong, Y.~Chen, and Y.~Zhang, ``Fast fluid antenna multiple access enabling massive connectivity,'' {\em IEEE Commun. Lett.}, vol.~27, no.~2, pp.~711--715, 2023.

\bibitem{WongsFAMA}
K.-K. Wong, D.~Morales-Jimenez, K.-F. Tong, and C.-B. Chae, ``Slow fluid antenna multiple access,'' {\em IEEE Trans. Wireless Commun.}, vol.~71, no.~5, pp.~2831--2846, 2023.

\bibitem{sFAMArefXL}
K.-K. Wong and et~al, ``Extra-large mimo enabling slow fluid antenna massive access for millimeter-wave bands,'' {\em IET Ele Lett.}, vol.~58, no.~25, pp.~1016--1018, 2022.

\bibitem{NoorsFAMA}
N.~Waqar, K.-K. Wong, K.-F. Tong, A.~Sharples, and Y.~Zhang, ``Deep learning enabled slow fluid antenna multiple access,'' {\em IEEE Commun. Lett.}, vol.~27, no.~3, pp.~861--865, 2023.

\bibitem{WongCuma}
K.-K. Wong, C.-B. Chae, and K.-F. Tong, ``Compact ultra massive antenna array: A simple open-loop massive connectivity scheme,'' {\em IEEE Trans. Wireless Commun.}, pp.~1--1, 2023.

\bibitem{Ref3}
M.~Khammassi, A.~Kammoun, and M.-S. Alouini, ``A new analytical approximation of the fluid antenna system channel,'' {\em IEEE Trans. Wireless Commun.}, pp.~1--1, 2023.

\bibitem{PLS1}
B.~Tang, H.~Xu, K.-K. Wong, K.-F. Tong, Y.~Zhang, and C.-B. Chae, ``Fluid antenna enabling secret communications,'' {\em IEEE Commun. Lett.}, vol.~27, no.~6, pp.~1491--1495, 2023.

\bibitem{PLS2Far}
F.~{Rostami Ghadi}, K.-K. {Wong}, F.~J. {Lopez-Martinez}, W.~K. {New}, H.~{Xu}, and C.-B. {Chae}, ``Physical layer security over fluid antenna systems,'' {\em arXiv:2402.05722v1}, pp.~1--11, 2024.

\bibitem{Sanchez}
J.~D. Vega-Sánchez, L.~F. Urquiza-Aguiar, H.~R.~C. Mora, N.~V.~O. Garzón, and D.~P.~M. Osorio, ``Fluid antenna system: Secrecy outage probability analysis,'' {\em IEEE Trans. Veh. Technol.}, vol.~73, no.~8, pp.~11458--11469, 2024.

\bibitem{Gradshteyn}
I.~S. Gradshteyn and I.~M. Ryzhik, {\em Table of Integrals, Series and Products}.
\newblock San Diego, CA, USA: Academic Press, 7~ed., 2007.

\bibitem{Perim}
V.~{Perim}, J.~D.~V. {Sánchez}, and J.~C. S.~S. {Filho}, ``Asymptotically exact approximations to generalized fading sum statistics,'' {\em IEEE Trans. Wireless Commun.}, vol.~19, no.~1, pp.~205--217, 2020.

\bibitem{Barros}
M.~{Bloch} and et~al, ``Wireless information-theoretic security,'' {\em IEEE Trans. Inf. Theory}, vol.~54, no.~6, pp.~2515--2534, 2008.

\end{thebibliography}


\end{document}